\documentclass{article}

\usepackage{amsmath,amssymb,amsfonts,amsthm}
\usepackage[numbers,sort&compress]{natbib}
\usepackage{color}
\usepackage{algorithm2e}
\usepackage[noend]{algorithmic}
\usepackage{graphicx}
\usepackage{verbatim}
\usepackage{subfig}	
\usepackage{setspace}
\usepackage{epsfig,color}
\usepackage[nice]{nicefrac}
\usepackage{breqn}
\usepackage{hyperref}
\usepackage{xspace}
\usepackage{url}
\usepackage{lineno}
\usepackage{multirow}
\usepackage{hhline}

\newtheorem{corollary}{Corollary}
\newtheorem{theorem}{Theorem}
\newtheorem*{helly}{Helly's Theorem}
\newtheorem*{fractional}{Fractional Helly Theorem}

\begin{document}

\title{A Note on Testing Intersection of Convex Sets \\ in Sublinear Time}
\author{Israela Solomon}

\maketitle

We present a simple sublinear time algorithm for testing the following geometric property. Let $P_1, ..., P_n$ be $n$ convex sets in $\mathbb{R}^d$ ($n \gg d$), such as polytopes, balls, etc. We assume that the complexity of each set depends only on~$d$ (and not on the number of sets $n$). We test the property that there exists a common point in all sets, i.e. that their intersection is nonempty. Our goal is to distinguish between the case where the intersection is nonempty, and the case where even after removing many of the sets the intersection is empty. In particular, our algorithm returns PASS if all of the $n$ sets intersect, and returns FAIL with probability at least $1-\epsilon$ if no point belongs to $\frac{\alpha}{d+1} n$ sets, for any given $0 < \alpha, \epsilon < 1$.

We call a collection of $q$ convex sets $P_{i_1},...,P_{i_q}$ a '$q$-tuple' (the order of elements in the tuple is not important).

The motivation for our algorithm comes from the following well known theorem by Eduard Helly \cite{Helly1923}\cite{HellyEnglish}\cite{HellyWiki}.

\begin{helly}
    Let $X_1, ..., X_n$ be a finite collection of convex sets in $\mathbb{R}^d$, with $n>d$. The intersection of every $d+1$ of these sets is nonempty, if and only if the whole collection has a nonempty intersection.
\end{helly}

Helly's theorem inspires us to consider the following algorithm.

\: \:
\RestyleAlgo{boxruled}
\begin{algorithm}[H]
\DontPrintSemicolon
\caption{Tester for nonempty intersection}

    \For {$i = 1, 2, ..., t$} {
        Pick $(d+1)$-tuple $P_{i_1}, P_{i_2}, ..., P_{i_{d+1}}$ with uniform distribution \;
        \If {$\cap_{j=1}^{d+1} P_{i_j} = \emptyset$}
            {\Return FAIL}
    }
    \Return PASS \;
\end{algorithm}
\: \: \:

\begin{theorem} \label{thm_correctness}
	For $t = \log _\alpha \epsilon$, the algorithm returns PASS if the intersection of all sets is nonempty, and returns FAIL with probability at least $1-\epsilon$ if the intersection of every $\frac{\alpha}{d+1} n$ sets is empty. The running time of the algorithm does not depend on $n$.
\end{theorem}

In order to prove the theorem, we should show that if any $\frac{\alpha}{d+1} n$ sets have an empty intersection, then the probability of picking a $(d+1)$-tuple whose intersection is empty and returning FAIL is high. For that purpose, we use the following theorem by Katchalski and Liu~\cite{FractionalHelly}.

\begin{fractional}
	Let $q>d>0$ be integers and let $F$ be a family of $n$ convex sets in $\mathbb{R}^d$ with $\alpha \binom{n}{q}$ of the $q$-tuples intersecting. Then there is a point in at least $\beta n$ of the sets, where $\beta = \left(\frac{\alpha}{\binom{q}{d}}\right) ^ {\frac{1}{q-d}}$.
\end{fractional}

Setting $q = d+1$, we obtain the following corollary.

\begin{corollary} \label{corollary_fractional}
	If the maximal number of sets from $P_1,...,P_n$ that intersect is less than $\frac{\alpha}{d+1} n$, then less than $\alpha \binom{n}{d+1}$ of the $(d+1)$-tuples intersect. Namely, at least $(1-\alpha) \binom{n}{d+1}$ of the $(d+1)$-tuples have empty intersection.
\end{corollary}

Now we are ready to prove the correctness of the algorithm.

\begin{proof}[Proof of Theorem \ref{thm_correctness}]
	If the intersection of all of the $n$ sets is nonempty, then the intersection of every $d+1$ sets is nonempty, and the algorithm returns PASS.
	
	Assume that the intersection of every $\frac{\alpha}{d+1} n$ sets is empty. Then, by Corollary~\ref{corollary_fractional}, there are less than $\alpha \binom{n}{d+1}$ choices of $d+1$ sets such that their intersection is nonempty, which means that the probability to pick such sets (and not returning FAIL) in each round is smaller than $\alpha$. Therefore, the probability that our algorithm returns FAIL is at least $1-\alpha ^ t$. Setting $t = \log _\alpha \epsilon$, the probability to return FAIL is at least $1-\epsilon$.
	
	The running time of the algorithm is $\log _\alpha \epsilon$ times the time it takes to compute the intersection of $d+1$ convex sets. Note that in general the complexity of a convex set is not bounded. Since we assume that the complexity of each set depends only on $d$, then the running time is $O(f(d) \cdot \log _\alpha \epsilon)$, which does not depend on $n$.
\end{proof}

\section*{Acknowledgments}
I would like to thank Ronitt Rubinfeld for her interest and for suggesting me to write this note.

\bibliographystyle{plainnat}
\bibliography{bibliography}

\end{document}